\documentclass[runningheads]{llncs}

\usepackage{booktabs} 
\usepackage[ruled]{algorithm2e} 

\SetAlFnt{\small}
\SetAlCapFnt{\small}
\SetAlCapNameFnt{\small}
\SetAlCapHSkip{0pt}
\IncMargin{-\parindent}

\usepackage{tikz}
\usepackage{physics}
\usepackage{xcolor}
\usepackage[outline]{contour} 
\contourlength{1.0pt}

\tikzset{>=latex} 
\colorlet{myred}{red!85!black}
\colorlet{myblue}{blue!80!black}
\colorlet{mydarkred}{myred!80!black}
\colorlet{mydarkblue}{myblue!60!black}
\tikzstyle{xline}=[myblue,thick]
\tikzstyle{yline}=[red,thick]
\tikzstyle{anotline}=[thin,->]
\def\tick#1#2{\draw[thick] (#1) ++ (#2:0.09) --++ (#2-180:0.18)}
\tikzstyle{myarr}=[myblue!50,-{Latex[length=3,width=2]}]

\usepackage[sectionbib]{natbib}
\setcitestyle{numbers,square}
\usepackage{graphicx}
\usepackage{mathrsfs}
\usepackage{setspace}

\usepackage{amsthm}
\usepackage{amssymb}
\usepackage{amsmath}

\usepackage[misc]{ifsym}
\usepackage{enumitem}
\usepackage{subfig} 
\usepackage{color}
\usepackage{cleveref} 
\usepackage{booktabs}
\usepackage{multirow}
\usepackage[title]{appendix}
\newtheorem{thm}{Theorem}
\newtheorem{defi}{Definition}
\newtheorem{lem}[thm]{Lemma}
\newtheorem{pro}{Proposition}
 
\newtheorem{exmp}{Example}

\DeclareMathOperator*{\argmax}{arg\,max}

\newcommand{\rev}{\mathsf{rev}}
\newcommand*\diff{\mathop{}\!\mathrm{d}}
\newcommand{\myerson}{\widetilde{p}}
\newcommand{\mmax}{\mathsf{max}}

\begin{document}
\title{Auction Design for Bidders with Ex Post ROI Constraints}

\author{Hongtao Lv\inst{1,3} \and
Xiaohui Bei\inst{2} \and  
Zhenzhe Zheng
\inst{3}\thanks{Z. Zheng is the corresponding author.} \and
Fan Wu
\inst{3}}
\authorrunning{H. Lv et al.}
%
\institute{School of Software \& Joint SDU-NTU Centre for Artificial Intelligence Research (C-FAIR), Shandong University, Jinan, China \\
\email{lht@sdu.edu.cn}
\and School of Physical and Mathematical Sciences,
Nanyang Technological University, Singapore, Singapore\\
\email{xhbei@ntu.edu.sg}
\and
Department of Computer Science and Engineering, Shanghai Jiao Tong University, Shanghai, China\\
\email{\{zhengzhenzhe, wu-fan\}@sjtu.edu.cn}
}

\maketitle              
\begin{abstract}
Motivated by practical constraints in online advertising, we investigate single-parameter auction design for bidders with constraints on their Return On Investment (ROI) -- a targeted minimum ratio between the obtained value and the payment. We focus on \emph{ex post} ROI constraints, which require the ROI condition to be satisfied for every realized value profile. 
With ROI-constrained bidders, we first provide a full characterization of the allocation and payment rules of dominant-strategy incentive compatible (DSIC) auctions. In particular, we show that given any monotone allocation rule, the corresponding DSIC payment should be the Myerson payment with a \emph{rebate} for each bidder to meet their ROI constraints.
Furthermore, we also determine the optimal auction structure when the item is sold to a single bidder under a mild regularity condition. This structure entails a randomized allocation scheme and a first-price payment rule, which differs from the deterministic Myerson auction and previous works on ex ante ROI constraints. 

\end{abstract}

\section{Introduction}\label{sec:intro}
\setcounter{footnote}{0}

Online advertising auctions are a vital source of revenue for many IT companies, generating billions of dollars of revenue annually.
In recent years, with tens of millions of ad auctions being conducted in real-time each day, this large-scale and complex market has prompted modern online advertising platforms to develop auto-bidding services, which allow the advertisers to set up high-level marketing goals for their ad campaigns and then bid on behalf of the advertisers.

In these auto-bidding scenarios, advertisers' financial constraints such as budget and return on investment (ROI) constraints have become critical in auction design. While auctions for budget-constrained bidders have been extensively studied in the literature~\cite{malakhov2008optimal, borgs2005multi, laffont1996optimal}, research on auction design for bidders with ROI constraints is still in its nascent stage.
The ROI constraints of advertisers require that the payment cannot be more than a certain fraction of the obtained advertising value. In other words, there is a targeted minimum ratio between the obtained value and the payment for an ROI-constrained bidder. Unlike budget constraints which set a hard limit on payment, ROI constraints establish a payment limit that is linearly related to the allocated value. Previous studies~\cite{golrezaei2021auction, auerbach2008empirical} have demonstrated that ROI constraints align better with real-world empirical evidence than budget constraints, and it is the aim of this paper to explore how to design auctions with good incentive and revenue guarantees for ROI-constrained bidders.

The existing literature on auction design for ROI-constrained bidders primarily focuses on \emph{ex ante} ROI constraints, which requires an \emph{expected} ROI with respect to the prior value distributions of bidders~\cite{golrezaei2021auction, balseiro2021landscape}.  This approach is suitable for advertisers who participate in a large number of auctions daily and are only concerned with their average spend per unit of value.
However, in reality, most ad campaigns experience the ``long-tail phenomenon''~\cite{brynjolfsson2011goodbye}, which means they only receive dozens of or fewer clicks per day. Under these conditions, an auction with ex ante ROI guarantees may have a non-negligible probability of violating the ROI constraints of these ad campaigns over a day. Due to these reasons, in this work, we focus on the \emph{ex post} or \emph{hard} ROI constraints, which ensures that the auction respects the ROI constraints of bidders for any realized value profile. This is a stronger requirement compared to ex ante ROI constraints and addresses the limitation of current auction design methods.

\subsection{Our Results} 
In this work, we examine the design of truthful and optimal auction design for ex post ROI-constrained bidders.  We inherit the setting from the classic single-parameter mechanism design and consider the values of bidders as private information and the targeted ROIs as public information. In this single-parameter environment, the ROI constraints can be integrated into the objective function (see Section~\ref{sec:prel} for details), resulting in a transformed utility model: $u_i = M_i v_i x_i - p_i$, where $u_i$ represents the utility of bidder $i$, $v_i$ is the value, $x_i$ is the allocation quantity, and $p_i$ is the payment. Here $M_i>1$ is the targeted ROI ratio, which differentiates this model from the classical quasilinear utility model. 

We first study the characterizations of truthful auctions with ROI-constrained bidders. Compared to Myerson's characterization of truthful auctions in the single-parameter environment, we show that the monotonicity requirement of the allocation rule remains true for ex post ROI-constrained bidders, but the unique payment rule in~\cite{myerson1981optimal} should be modified by subtracting a $\max$ term, which can be interpreted as a ``rebate'' equal to the largest ``violation'' of the Myerson payment to the ROI constraint for all lower valuations. 
This is a full characterization that completely describes all truthful auctions with ROI-constrained bidders. This result can be proved using similar techniques from Myerson's analysis. It can also be derived from the following alternative interpretation of the payment rule: note that the ROI-constrained bidder assigns a weight $M_i > 1$ to her obtained value $v_ix_i$ from the allocation, but not to her payment. To not violate the individual rationality (IR) condition, instead of applying a naive approach that charges the bidder $M_i$ times the Myerson payment, we must iteratively apply the Myerson payment increment (multiplied by $M_i$) in small intervals and truncate the payment at the obtained value whenever necessary.

Next, we turn our focus to the optimal (\emph{i.e.} revenue-maximizing) auction design. The additional $\max$ term in our payment rule poses a significant challenge to the optimal auction design, since it is unclear how this term can be incorporated into a modified virtual valuation function as seen in previous literature. Instead, we concentrate on the case of selling a single item to a single bidder. Our main result suggests that under a mild regularity assumption known as \emph{decreasing marginal revenue} (DMR)\footnote{DMR requires the marginal revenue, $vf(v)+F(v)-1$ to be non-decreasing in the \emph{value space}. This is different from the usual definition of regularity which requires the same monotonicity but in the \emph{quantile space}. Please see the related work section for a more detailed discussion of their differences and more related works.}, the optimal auction for selling to a single ex post ROI-constrained bidder employs a randomized scheme.
More specifically, the allocation rule $x(\cdot)$ starts with a \emph{first-price} interval, where the payment always matches the obtained value, until it reaches the highest allocation and $x(\cdot)$ becomes constant thereafter. This finding is in contrast to the classic Myerson auction~\cite{myerson1981optimal} and previous results for bidders with ex ante ROI constraints, where the optimal auctions are always deterministic. It implies that similar to much literature on optimal mechanism design for multi-parameter settings, a slight generalization, such as the inclusion of the $M_i$ term, in the single-parameter setting can lead to randomized optimal auctions.


\subsection{Related work}
There are two main threads of studies of auctions with ROI-constrained bidders. The first thread investigates how the bidding strategies of the bidders are affected by the ROI constraints in classic VCG or generalized second price (GSP) auctions~\cite{szymanski2006impact, borgs2007dynamics, golrezaei2021bidding, tillberg2020optimal,  aggarwal2019autobidding, heymann2019cost, babaioff2021non}.
The second thread, which our paper follows, focuses on the \emph{design} of auctions with ROI-constrained bidders, which is of practical interest to many online advertising platforms~\cite{golrezaei2021auction, balseiro2021landscape, golrezaei2021bidding, li2020incentive, cavallo2017sponsored, wilkens2017gsp, li2022auto, balseiro2019black, mehta2022auction}.  In this line of study, the most related work to ours is \cite{golrezaei2021auction}, which showed empirically that a fraction of the buyers in online advertising are indeed ROI-constrained. They also took the first step towards revenue-maximizing auction design for bidders with \emph{ex ante} ROI constraints.
Note that ex ante ROI constraints only require the ROI conditions to be met in expectation and are strictly weaker than ex post ROI constraints.
Another recent work \cite{balseiro2021landscape} considered the scenario where either the value or the ROI constraint is private information of the bidder. They used a similar utility function as ours, but still focus on the concept of ex ante ROI. Unlike these works, we concentrate on \emph{ex post} ROI constraints, which provide a hard ROI guarantee for bidders in every possible value realization. 
In~\cite{deng2021towards, balseiro2021robust}, 
the authors
considered  ex post ROI constraints in multiple stages, and assumed that each bidder maintains a fixed bid multiplier among stages, which leads to completely different problems from ours. 

One particular line of research focuses on requirements of truthfulness for ex post ROI constraints. Cavallo \emph{et al.} studied the same utility function as ours in~\cite[Appendix A]{cavallo2017sponsored}, and investigated the corresponding payment rules. 
The main difference is that, they limited their focus on \emph{deterministic} mechanisms for bidders with \emph{identical} ROI constraints, while we consider a more general single-parameter setting in the randomized mechanism domain.   
Li \emph{et al.}~\cite{li2020incentive} proposed a condition on truthfulness of the ROI information, based on which they provided a mechanism framework using tools from control theory. They took the ROI constraints as private information, instead of the value, which leads to a substantially different problem from ours. 

The DMR assumption used in our optimal auction characterization has been widely discussed in the literature. It means that the function  $\psi(v) \triangleq  v f(v) + F(v) - 1$ is non-decreasing, or equivalently $v\cdot (1 - F(v))$, which is the expected revenue of selling the item at price $p$, is concave, and this is where the name of this condition comes from. Intuitively, many commonly used distribution functions 
satisfy this assumption, \emph{e.g.}, uniform distributions, and any
distribution of finite support and monotone non-decreasing density. 
The DMR condition was first proposed in~\cite{che1998standard} for bidders with budget constraints. In~\cite{fiat2016fedex}, the authors found that the DMR condition is more natural in their setting than the traditionally used notion of \emph{regularity}~\cite{myerson1981optimal}, since DMR precisely removes the requirement of ironing in the \emph{value} space, instead of in the \emph{quantile} space as in~\cite{myerson1981optimal}. In~\cite{devanur2017optimal2}, DMR was discussed comprehensively, and the authors showed that the optimal mechanism is deterministic under the DMR condition in  a multi-unit setting with private demands. We refer the reader to their work for concrete examples and more discussion.


\section{Preliminaries}
\label{sec:prel}
We consider a general single-parameter auction environment, which consists of a seller and $n$ bidders $ \boldsymbol{N}=\{1,2,...,n\}$. Each bidder $i$ has a private valuation $t_i$ per	unit of	the	good. We represent $x_i$ as the quantity of the allocated good to bidder $i$ and $p_i$ as the payment of bidder $i$. Without loss of generality, we assume the maximum possible allocation is ${x}_i^{\mmax}= 1$ and the good is indivisible, that is, $x_i$ denotes the probability of bidder $i$ receiving the good.
Besides the allocated value, each bidder also has a return on investment (ROI) constraint $M_i$, as public information\footnote{This setting is practical and prevalent in practice, \emph{e.g.}, in online advertising, the targeted ROI typically remains the same over a certain period.},
which specifies the minimum targeted ratio between her obtained value and the payment. We assume $1< M_i< +\infty$ in this work. We note that the ROI constraint is considered in an ex post measure, \emph{i.e.}, it requires that $\frac{t_i x_i}{p_i}\geq M_i$ strictly holds in the outcome of every auction instance. Note that the same model is also adopted in \cite[appendix A]{cavallo2017sponsored}.

With the above definitions, the utility of bidder $i$ is given by
\begin{equation}
\label{utility-pre}
    u_i=\left\{
    \begin{array}{lr}
    t_i x_i - p_i  & \quad \text{if } \frac{t_i x_i}{p_i}\geq M_i \\
    -\infty & \quad \textrm{otherwise}.
    \end{array}
\right.
\end{equation}
It is worth noting that this is the standard quasilinear utility model with the addition of the ROI constraint.
We can further define 
$$v_i = \frac{t_i}{M_i},$$
which could be interpreted as the maximum willingness-to-pay of the bidder $i$ per unit of the good. Then, we can rewrite the utility function as 
\begin{equation}
\label{utility}
    u_i=\left\{
    \begin{array}{lr}
    M_i v_i x_i - p_i  & \quad \text{if } \, v_i x_i\geq p_i \\
    -\infty & \quad \textrm{otherwise}.
    \end{array}
\right.
\end{equation}
One can observe that, as $M_i$ is a public constant, $v_i$ and $t_i$ are completely interchangeable.
To avoid confusion, we use the term \emph{value} to  represent $v_i$, and \emph{initial value} to represent $t_i$ in the following discussion. Each value $v_i$ is independently drawn from a probability distribution $F_i: [0,v_{\mmax}]\rightarrow [0,1]$, with a continuous probability density function $f_i$. While the distributions $F_i$'s are common knowledge, the exact value $v_i$ is known only to the bidder $i$. We denote $\textbf{v}$ as the value profile of all bidders, and $\textbf{v}_{-i}$ as that of all bidders except bidder $i$. 

In an auction, each bidder reports her value as $b_i$, which is not necessarily equal to $v_i$. We define $\textbf{b}$ and $\textbf{b}_{-i}$ similarly as the notations of $\textbf{v}$ and $\textbf{v}_{-i}$. 
Based on the reported bids, an auction mechanism consists of an allocation rule $x_i(b_i, \textbf{b}_{-i})$, mapping the bid profile to the allocated quantity to each bidder $i$, and a payment rule $p_i(b_i, \textbf{b}_{-i})$, mapping the bid profile to the payment for each bidder. When clear from context, we will omit $\textbf{b}_{-i}$ in the mappings. We also use $u_i(b_i, v_i)$ to represent the utility of bidder $i$ who has value $v_i$ and bid $b_i$.
In the following discussion, we assume that the allocation rule $x_i(\cdot)$ is always right-differentiable, and there are finite non-differentiable points. When $x_i(\cdot)$ is non-continuous at $v$, let $x_i(v) = \lim_{z\rightarrow v^+} x_i(z)$.

We are interested in auctions that are \emph{dominant-strategy incentive compatible} (DSIC) and \emph{individually rational} (IR).
\begin{defi}[Dominant-Strategy Incentive Compatibility, DSIC]
A mechanism is dominant-strategy incentive compatible if and only if
$$u_i(v_i, v_i)\geq u_i(b_i, v_i), \quad \forall b_i, \mathbf{b}_{-i}, i\in  \boldsymbol{N}.$$
\end{defi}
\begin{defi}[Individual Rationality, IR]
A mechanism is individually rational if and only if 
$$u_i(v_i, v_i)\geq 0,\quad \forall v_i, \mathbf{b}_{-i}, i\in  \boldsymbol{N}.$$
\end{defi}
For ease of notation, we use \emph{truthfulness} to represent the properties of both DSIC and IR in the following sections.
In addition, for truthful auctions, we do not distinguish $v_i$ and $b_i$ hereinafter.

The revenue of a truthful auction is defined as
$$\rev =\mathbb{E}_{\textbf{v}}\left[\sum_{i\in  \boldsymbol{N}}  p_i(v_i) \right].$$
The aim of this work is to characterize both truthful and revenue-maximizing (optimal) auctions with ex post ROI-constrained bidders.

\section{Characterize the Structure of DSIC Auctions}
\label{sec3}
In this section, we present characterizations of the DSIC auctions with ex post ROI constraints.
These results generalize the classical Myerson's Lemma~\cite{myerson1981optimal} for the traditional utility model (\emph{i.e.}, $M_i=1$), which states that 
in the single-parameter environment, a mechanism is DSIC if and only if its allocation rule is monotone and the payment scheme follows a unique rule.

\begin{lem}[Myerson’s Lemma~\cite{myerson1981optimal}]
\label{Myerson}
For traditional bidders with $M_i=1$, a single-parameter mechanism is DSIC if and only if:
\begin{itemize}
    \item {[Monotone Allocation Rule]} the allocation rule is monotonically non-decreasing, \emph{i.e.}, $x_i(v)\leq x_i(v')$ for all $v<v'$ and bidder $i$;
    \item {[Unique Payment Rule]} for each monotonically non-decreasing allocation rule $x_i(\cdot)$, and $p_i(0)=0$, the payments are given by 
    \begin{equation}
    \label{eq:myerson}
     p_i(v)=vx_i(v)- \int_0^{v}x_i(z)\diff z.
    \end{equation}
\end{itemize}
\end{lem}

Clearly, these results cannot be directly applied to the ROI-constrained bidders, because the payment derived from Myerson's Lemma may violate the ROI constraints. 
The main result in this section is a complete characterization of the DSIC mechanisms with ROI-constrained bidders. We will see that the monotonicity condition for the allocation remains the same, but the payment rule needs to be modified appropriately to accommodate the ROI constraints.

\begin{thm}[Characterization]
\label{thm:characterization}
For ex post ROI-constrained bidders, a single-parameter mechanism is DSIC if and only if:
\begin{itemize}
    \item {[Monotone Allocation Rule]} the allocation rule is monotonically non-decreasing, \emph{i.e.}, $x_i(v)\leq x_i(v')$ for all $v<v'$ and bidder $i$;
    \item {[Unique Payment Rule]} for each monotonically non-decreasing allocation rule $x_i(\cdot)$, and $p_i(0)=0$, the payments are given by
    \begin{equation}
        \label{eq:payment}
        p_i(v) = M_i \myerson_i(v) - \max_{0 \leq z \leq v}\{M_i \myerson_i(z) - z x_i(z)\},
    \end{equation}
    where $\myerson_i$ is the Myerson payment  given in~\eqref{eq:myerson}.
\end{itemize}
\end{thm}
We can interpret this characterization from two perspectives: First, from the perspective of the initial value $t_i$ of bidder $i$, it suggests that compared to the classic Myerson auction, an ROI-constrained bidder with value $v$ will need to pay the initial Myerson payment $M_i \myerson_i(v)$ (recall that $t_i=M_iv_i$), minus a ``rebate'' which equals to the largest ``violation'' of the Myerson payment to the ROI requirement when the bidder's valuation is no more than $v$.

Second, from the perspective of the value $v_i$, since the ROI-constrained bidder assigns a weight $M_i > 1$ to her obtained value from the allocation, but not to her payment, a naive application that charges the bidder $M_i$ times of the Myerson payment may violate the IR constraint. Therefore, we need to iteratively apply the Myerson payment increment (multiplied by $M_i$) in small intervals and truncate the payment at the value whenever necessary.
These two perspectives are mathematically equivalent, and we adopt the second perspective in the following for exposition convenience.

Next, before proving this theorem, some observations are immediate from this characterization. We defer the proof of these observations to Appendix~\ref{appen:prop_1}.
\begin{pro}
\label{prop:observation}
\ 
\begin{enumerate}
    \item We always have $p_i(v) \leq vx_i(v)$ and $p_i(v) \leq M_i \myerson_i(v)$ for any bidder $i$ and value $v$ in DSIC mechanisms.
    \item The payment function $p_i(\cdot)$ is monotonically non-decreasing for any bidder $i$ in DSIC mechanisms.
\end{enumerate}
\end{pro}

Now we proceed to prove Theorem~\ref{thm:characterization}. The proof consists of showing the following claims in sequence. It is not difficult to see that these three claims together imply Theorem~\ref{thm:characterization}.

\begin{enumerate}
    \item For ex post ROI-constrainted bidders, if a mechanism is DSIC, then the allocation rule must be monotonically non-decreasing.
    \item Any monotonically non-decreasing allocation rule $x_i(\cdot)$ with the payment rule given in~\eqref{eq:payment} produces a DSIC mechanism. 
    \item Given any monotone allocation rule $x(\cdot)$, the payment rule $p(\cdot)$ such that $(x, p)$ is DSIC, if exists, must be unique.
\end{enumerate}

The analyses of steps (1) and (3) are very similar to the proof of the original Myerson's Lemma, and we defer the details to Appendix~\ref{appen:step_1} and~\ref{appen:step_3}. 
Next, we prove step (2). 
When clear from context, we will drop the subscript $i$ in $x_i(\cdot), p_i(\cdot), u_i(\cdot)$ and $M_i$ as shorthand in the following proofs.

\begin{proof}[Proof of Step (2)]
    Consider a bidder $i$ with private valuation $v$ and fix the other bids \textbf{$b_{-i}$}. 
    We examine the utilities of bidder $i$ when she bids her true valuation and when she bids some different value $v' \neq v$. Consider two cases.
    \begin{itemize}
        \item When $v' < v$, we have $\max_{0 \leq z \leq v}\{M \myerson(z) - z x(z)\} \geq \max_{0 \leq z \leq v'}\{M \myerson(z) - z x(z)\}$, which implies
        $$p(v) - p(v') \leq M(\myerson(v) - \myerson(v')).$$
        This inequality effectively removes the $\max$ term in the payment formula~\eqref{eq:payment} and reduces the problem to that with the Myerson payment. This allows us to apply the standard argument for the Myerson auction to show the DSIC property of our mechanism. We show the analysis below for completeness.
        
        We can compute the utility difference of bidder $i$ when she bids $v$ and $v'$, and get
        \[
        \begin{aligned}
          u(v, v) - u(v', v) &= \left(M v x(v) - p(v)\right) - \left(M v x(v') - p(v')\right) \\
          &\geq M(vx(v) - \myerson(v)) - M(vx(v') - \myerson(v')) \\
          &= M \int_0^v{x(z)}\diff z - M\left(vx(v') - v'x(v') + \int_0^{v'}{x(z)}\diff z\right)\\
          &=M\left(\int_{v'}^v{x(z)}\diff z - (v-v')x(v')\right) \\
          &\geq M\left(\int_{v'}^v{x(v')}\diff z - (v-v')x(v')\right) = 0,
        \end{aligned}
        \]
        where the second equality is by plugging in the Myerson payment formula~\eqref{eq:myerson}. This means bidder $i$ has no incentive to misreport her valuation $v$ as $v'$ in this case.
        \item When $v' > v$, we examine $\max_{0 \leq z \leq v}\{M \myerson(z) - z x(z)\}$ and $\max_{0 \leq z \leq v'}\{M \myerson(z) - z x(z)\}$. There are two possibilities:
        \begin{itemize}
            \item If these two terms are equal, then we can apply the same argument as in the previous case (and also as in the Myerson auction analysis) to prove the DSIC property. We omit the details here.
            \item If $\max_{0 \leq z \leq v}\{M \myerson(z) - z x(z)\} < \max_{0 \leq z \leq v'}\{M \myerson(z) - z x(z)\}$, this means $\argmax_{0 \leq z \leq v'}\{M \myerson(z) - z x(z)\} = v^* > v$. Then at valuation $v'$, we should have 
            \[
            \begin{aligned}
              p(v') &= M\myerson(v') - (M \myerson(v^*) - v^*x(v^*))\\
              &= M\left(v'x(v') - v^*x(v^*) - \int_{v^*}^{v'}{x(z)}\diff z\right) + v^* x(v^*) \\
              (\textrm{replace $M$ by 1}) &\geq v'x(v') - \int_{v^*}^{v'}{x(z)}\diff z \\
              &\geq v'x(v') - \int_{v^*}^{v'}{x(v')}\diff z \\
              &= v'x(v') - (v'-v^*)x(v') \\
              &= v^*x(v') > vx(v').
            \end{aligned}
            \]
            That is to say, when reporting $v'$, the payment of bidder $i$ will be greater than the value she obtains (which is $vx(v')$), therefore violating the IR condition. So bidder $i$ has no incentive to misreport as $v'$ in this case.
        \end{itemize}
    \end{itemize}
    \end{proof}

\section{Optimal Auction Design for a Single Bidder} 
\label{sec:opt}
Having obtained the precise characterization of the allocation rule and payment function in the setting with ROI constraints, we now turn to the revenue maximization auction design.
Recall that in the Myerson auction~\cite{myerson1981optimal} and previous works in the ex ante ROI constraints setting~\cite{balseiro2021landscape, golrezaei2021auction}, the revenue maximization problem is reduced to the problem of maximizing (modified) virtual welfare. Unfortunately, with ex post ROI constraints, the payment function characterization~\eqref{eq:payment} involves an additional $\max$ term compared to the Myerson payment, and it is unclear how to incorporate this term into a modified virtual valuation formulation. We present the following simple example with a single bidder to demonstrate that, unlike the Myerson auction, the allocation that maximizes the virtual welfare may no longer be optimal with ROI-constrained bidders.
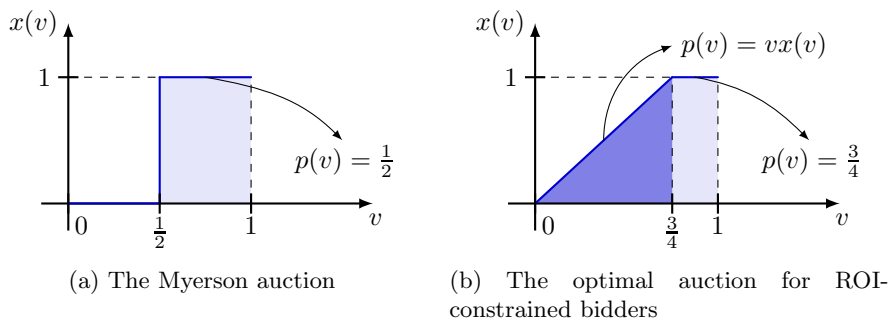
\begin{figure}
\centering
    \subfloat[The Myerson auction]{
    \label{fig:myerson}

\begin{tikzpicture}[scale=1.5]
\def\textscale{1}
\def\xmax{2.7} 
\def\ymax{1.6} 
  \message{^^JConstant}
  \def\kx{0.6*\xmax}
  \def\ky{0.7*\ymax} 
  \def\a{0.5*\kx} 
  \fill[myblue!10] (\a,0) rectangle (\kx,\ky);
  \draw[->,thick] (0,-0.15*\ymax) -- (0,\ymax) node[left,scale=\textscale] {$x(v)$};
  \draw[->,thick] (-0.15*\ymax,0) -- (\xmax,0) node[right=1,below,scale=\textscale] {$v$};
  \draw[xline,line cap=round] (0,0) -- (\a,0);
  \draw[xline,line cap=round] (\a,0) -- (\a,\ky);
  \draw[xline,line cap=round] (\a,\ky) -- (\kx,\ky);
  \draw[dashed] (0,\ky) -- (\a,\ky);
  \draw[dashed] (\kx,0) -- (\kx,\ky);
  \draw[anotline,line cap=round] (0.75*\kx,\ky) to [out=-10,in=135] (1.5*\kx,0.5*\ky) node[left=-1,below=-1.7,scale=\textscale] {$p(v)=\frac12$};
  \tick{\a,0}{90} node[below=-2.5,scale=\textscale] {$\frac{1}{2}$};
  \tick{\kx,0}{90} node[below=-2.5,scale=\textscale] {$1$};
  \tick{0,\ky}{0} node[left=-1,scale=\textscale] {$1$};
  \tick{0,0}{90} node[left=-5,below=-2.5,scale=\textscale] {$0$};  
\end{tikzpicture}
    }
    \hspace{.1in}
\subfloat[The optimal  auction for ROI-constrained bidders]{
\label{fig:example_rand}

\begin{tikzpicture}[scale=1.5]
\def\textscale{1}
  \def\xmax{2.7} 
  \def\ymax{1.6} 
  \def\kx{0.6*\xmax}
  \def\ky{0.7*\ymax} 
  \def\a{0.75*\kx} 
  \fill[myblue!10] (\a,0) rectangle (\kx,\ky);
  \fill[myblue!50] (0,0) -- (\a,\ky) -- (\a,0) -- cycle;
  \draw[->,thick] (0,-0.15*\ymax) -- (0,\ymax) node[left,scale=\textscale] {$x(v)$};
  \draw[->,thick] (-0.15*\ymax,0) -- (\xmax,0) node[right=1,below,scale=\textscale] {$v$};
  \draw[xline,line cap=round] (0,0) -- (\a,\ky);
  \draw[xline,line cap=round] (\a,\ky) -- (\kx,\ky);
  \draw[dashed] (0,\ky) -- (\a,\ky);
  \draw[dashed] (\kx,0) -- (\kx,\ky);
  \draw[dashed] (\a,0) -- (\a,\ky);
  \draw[anotline,line cap=round] (0.875*\kx,\ky) to [out=-10,in=135] (1.5*\kx,0.5*\ky) node[left=-1,below=-1.7,scale=\textscale] {$p(v)=\frac34$};
  \draw[anotline,line cap=round] (0.375*\kx,0.5*\ky) to [out=90,in=200] (0.75*\kx,1.25*\ky) node[right=0,scale=\textscale] {$p(v)=vx(v)$};
  \tick{\a,0}{90} node[below=-2.5,scale=\textscale] {$\frac34$};
  \tick{\kx,0}{90} node[below=-2.5,scale=\textscale] {$1$};
  \tick{0,\ky}{0} node[left=-1,scale=\textscale] {$1$}; 
  \tick{0,0}{90} node[left=-5,below=-2.5,scale=\textscale] {$0$};  
\end{tikzpicture}
}
\caption{The Myerson auction and the optimal auction for one bidder with uniform value distribution over $[0,1]$ and $M=2$.}
\label{fig:myerson_and_rand}
\end{figure}

\begin{exmp}
Consider selling a single item to a single bidder with ROI constraint $M=2$ and valuation for the item $v$ following a uniform distribution $U[0, 1]$. If we disregard the ROI constraint (\emph{i.e.}, let $M=1$), the virtual valuation of this bidder is $\phi(v) = 2v - 1$, and the optimal Myerson auction, as shown in Fig.~\ref{fig:myerson}, sells the item at price $p=\frac12$ with the expected revenue of $\frac12\cdot \frac12=\frac14$.

However, with the ROI constraint $M=2$ in presence, this allocation rule ($x(v)=0$ when $v < 1/2$ and $x(v)=1$ otherwise) is no longer optimal.
As shown in Fig.~\ref{fig:example_rand}, the optimal auction, which will be proved in Theorem~\ref{thm:opt} later in this section, is a randomized auction with the allocation rule given by
\begin{equation*}
x(v) = \left\{
\begin{array}{lr}
   \frac43 v & \quad \text{if } v \leq \frac34 
   \\
   \displaystyle
   1 & \quad \text{if } v > \frac34.
\end{array}
\right.
\end{equation*}
This allocation rule would generate an expected revenue of $\frac38$, which is higher than $\frac14$.
\end{exmp}

This example already highlights an important feature of the optimal auction with ROI constraints: the allocation and payment may be randomized, even in the simple setting with a single bidder and uniform value distribution. It also suggests that it is difficult to follow the Myerson auction regime and reduce the revenue maximization problem to a welfare maximization problem with some modified virtual valuation. It seems a very challenging problem to obtain a characterization for the optimal auction in this setting.
Instead, in this section we focus on the special case when the item is sold to a single bidder. As we will show in the following analysis, this is already a nontrivial and interesting problem to design an optimal auction for a single bidder.

First, we show that with a single ROI-constrained bidder, the $\max$ term in the payment formula~\eqref{eq:payment} would be always $0$, reducing the payment rule~\eqref{eq:payment} to the standard Myerson payment (multiplied by a factor of $M$). We defer the proof to Appendix~\ref{appen:lem_pricing_opt}.

\begin{lem}
\label{lem:pricing_opt}
In the optimal auction with a single ROI-constrained bidder, let $(x, p)$ be the revenue-maximizing auction, then we have $M \myerson(v) \leq v\cdot x(v)$ for any value $v \in [0,v_{\mmax}]$.
In other words, $\max_{0 \leq z \leq v} \{M \myerson(z) - z x(z)\} = 0$ (which is achieved at $z=0$) for all $v \in [0,v_{\mmax}]$, and the payment rule reduces to $p(v) = M \myerson(v)$.
\end{lem}

With Lemma~\ref{lem:pricing_opt} at hand, it seems with a single bidder, we are back to the classic Myerson regime, where the revenue maximization problem can be converted to a welfare maximization problem with respect to the virtual valuation. 
That is, recall from the Myerson's theorem~\cite{myerson1981optimal}, we have 
\[
    \rev = \int_0^{v_{\mmax}} \phi(v) x(v) f(v) \diff v,
\]
where $\phi(v) = \left(v-\frac{1-F(v)}{f(v)} \right)$ is the \emph{virtual valuation}.
However, we still have the additional constraint that $M \myerson(v) \leq v\cdot x(v)$ for every $v$. This restricts our allocation space and turns the problem into a constrained welfare maximization problem.

In the following, we provide some further characterizations on the structure of the optimal auction with a single ROI-constrained bidder. We defer the proof to Appendix~\ref{appen:lem_TwoPossibility}.

\begin{lem}
\label{lem:TwoPossibility}
With a single ROI-constrained bidder, there always exists a revenue-maximizing auction such that for any valuation $v$, at least one of the following statements holds:
\begin{itemize}
    \item the derivative of allocation rule at the valuation $v$ exists, and $x'(v)=0$;
    \item the payment follows the first-price rule, \emph{i.e.}, $p(v) = v x(v)$.
\end{itemize}
\end{lem}

Lemma~\ref{lem:TwoPossibility} allows us to focus on auctions with a very specific structure: as long as the allocation is not constant, it always follows the first-price payment rule. In particular, combining with Lemma~\ref{lem:pricing_opt}, it implies whenever $x'(v)$ exists and $x'(v)>0$, we always have $p = v x(v) = M \myerson(v)$.

Next, we want to obtain a further characterization of the optimal auction with a single bidder. However, to do this would require us to make a mild assumption on the value distribution of the bidder, which is known as the \emph{Decreasing Marginal Revenue (DMR)} condition by~\cite{devanur2017optimal2}.

\begin{defi}[Decreasing Marginal Revenue, DMR]
\label{def:DMR}
    The value distribution of a bidder satisfies the condition of decreasing marginal revenue  if and only if the function
    \begin{equation*}
    \psi(v) \triangleq \phi(v) f(v) = v f(v)+ F(v) - 1 
    \end{equation*}
    is monotonically non-decreasing.
\end{defi}

Note that $\psi(v)$ being non-decreasing is equivalent to the fact that $v\cdot (1 - F(v))$, which is the expected revenue of selling the item at price $p$, being concave, and this is where the name of this condition comes from. Intuitively, many commonly used distribution functions 
satisfy this assumption, \emph{e.g.}, uniform distributions, and any
distribution of finite support and monotone non-decreasing density. 
The DMR condition is closely related to the regularity condition but they are incompatible\footnote{The regularity condition is equivalent to the expected revenue being concave in the \emph{quantile} space, while the DMR condition means the expected revenue is concave in the \emph{value} space.}. We refer the reader to~\cite{devanur2017optimal2} for concrete examples and more discussion.

With the assumption of DMR, the optimal auction exhibits an even simpler structure than what is described in Lemma~\ref{lem:TwoPossibility}, namely that there exist only two intervals in the optimal auction: interval of $(0, D)$ with $x'(v) > 0$ and interval of $(D, v_{\mmax})$ with  $x'(v) = 0$, where $D$ is a threshold valuation between them. This leads to our main theorem in this section, which characterizes the optimal allocation rule and payment rule for a single ROI-constrained bidder.

\begin{thm}
\label{thm:opt}
The optimal auction for a single ex post ROI-constrained bidder with a DMR value distribution over $[0,v_{\mmax}]$ is as follows:  
\begin{itemize}
    \item when $v<D$, the allocation is given by $$x(v) = \left (\frac{v}{D} \right)^{\frac{1}{M-1}},$$ and the payment follows the first-price rule, \emph{i.e.}, $p(v) =v x(v)$; 
    \item when $v\geq D$,  the allocation rule is $x(v) = 1$, and the payment is given by $p(v) =D$.
\end{itemize}
Here $D$ is a threshold valuation given as follows:
\begin{itemize}
    \item if $\int_0^{v_{\mmax}} \psi(v) v^{\frac{1}{M-1}} \diff v>0$, then $D=D^{*}$ such that $\int_0^{D^{*}} \psi(v) v^{\frac{1}{M-1}} \diff v = 0$;
    \item if $\int_0^{v_{\mmax}} \psi(v) v^{\frac{1}{M-1}} \diff v\leq 0$, then $D=v_{\mmax}$.
\end{itemize}
\end{thm}

This theorem provides an important insight that the optimal auction in the ROI-constrained setting is a randomized mechanism. Note that Myerson's optimal auction in the single-parameter setting is deterministic, but a decent body of works has shown that many generalizations to multi-parameter settings will lead to randomized optimal auctions~\cite{pavlov2011optimal, hart2015maximal}. Theorem~\ref{thm:opt} indicates that, even in the single-parameter environment, a slight generalization with an ROI constraint to the bidder will also lead to a randomized optimal auction.

We prove the theorem via the following steps. First, we derive the allocation of an optimal auction in an interval $(0, v^*)$ when $x'(v)$ is always positive in that interval.
Then, we show in Lemma~\ref{lem:OneInterval} that there exist only two intervals in the optimal auction: interval of $(0, D)$ with $x'(v) > 0$ and interval of $(D, v_{\mmax})$ with  $x'(v) = 0$. Finally, we will compute the optimal threshold valuation $D$ between these two intervals.

\begin{pro}
\label{prop:ExactAllocation}
If for some $v^* \in (0, v_{\mmax}]$, we have $x'(v)>0$ for all $v\in (0,v^*)$ in an optimal auction, then $x(\cdot)$ is continuous at $v^*$, and the allocation rule $x(v)$ for all $v \in [0, v^*]$ is given as:
$$x(v) = \left(\frac{v}{v^*}\right)^{\frac{1}{M-1}}x(v^*).$$ 
\end{pro}
\begin{proof}
We first assume $x(\cdot)$ is continuous at $v^*$, and we will prove later that, if it is discontinuous, we can  improve the revenue without violating the DSIC property.
By Lemma~\ref{lem:pricing_opt} and Lemma~\ref{lem:TwoPossibility}, we get that for all valuations $v \in [0, v^*]$, $M \myerson(v) - v x(v)=0$ always holds, that is,
$$M \left(v x(v) - \int_0^{v} x(z) \diff z \right) - v x(v)=0.$$
After transposition and derivation, this translates to
$$x(v) - (M-1) v x'(v)=0.$$
By solving this differential equation with the value of $x(v^*)$ at valuation $v^*$, we can get 
\begin{equation}
    \label{eq:ContinuousAllocation}
    x(v) = \left(\frac{v}{v^*}\right)^{\frac{1}{M-1}}x(v^*), \quad \forall v \in [0, v^*].
\end{equation}
Next, if $x(\cdot)$ is discontinuous at $v^*$, we need to replace $x(v^*)$ in~\eqref{eq:ContinuousAllocation} with $x(v^{*-})$, \emph{i.e.}, the left limit of $x(\cdot)$ at $v^*$ (recall that we denote $x(v^*)$ as the right limit when it is discontinuous). Since $x(v^{*-})< x(v^*)$, we can observe that directly using~\eqref{eq:ContinuousAllocation} as the allocation rule will increase the revenue without violating the DSIC property, which also makes $x(\cdot)$ continuous at $v^*$. 
This concludes the proof.
\end{proof}

\begin{lem}
\label{lem:OneInterval}
For a single ROI-constrained bidder with DMR value distribution, if there exist intervals with $x'(v)=0$ in an optimal auction, then there is exactly one such interval, and it appears in the highest value region.
\end{lem}

\begin{proof}
    Assume by contradiction that there exist multiple intervals with $x'(v) = 0$ in an optimal auction. Pick $(\underline{v}, \bar{v})$ to be the first such interval. 
    That is, we have $x'(v) > 0$ for all $v \in (0, \underline{v})$, $x'(v) = 0$ for all $v \in (\underline{v}, \bar{v})$, and $\bar{v}<v_{\mmax}$.

    First, we must have $\underline{v} > 0$. That is, the allocation rule cannot start with a flat interval. To see why this is true, assume otherwise that $\underline{v} = 0$. We focus on a point $v' = \bar{v} + \delta$ in the next interval for a sufficiently small $\delta > 0$ such that $v' < M\bar{v}$. Then there are two cases: (1) $x'(v')=0$ and (2) $x'(v')$ is strictly positive. In the first case, we will have $x(\bar{v}) > 0$, and the Myerson price at $\bar{v}$ will be $\myerson(\bar{v}) = \bar{v} \cdot x(\bar{v}) < M \myerson(\bar{v})$, which directly contradicts Lemma~\ref{lem:pricing_opt}. In the second case, we look at the payment $p(v')$ at point $v'$. Note that since $x'(v) > 0$, by Lemma~\ref{lem:pricing_opt} and Lemma~\ref{lem:TwoPossibility}, we should have $M\myerson(v') = v'x(v')$. But this cannot happen because
    \begin{equation*}
        \begin{aligned}
        M\myerson(v') &= M \left(v' x(v') - \int_0^{v'} x(z) \diff z \right) = M \left(v' x(v') - \int_{\bar{v}}^{v'} x(z) \diff z \right)\\      
        &> M \left(v' x(v') - (v' - \bar{v}) x(v') \right) = M \bar{v} x(v') > v' x(v').
        \end{aligned}
    \end{equation*}
    
    Knowing $\underline{v} > 0$, by Proposition~\ref{prop:ExactAllocation}, we know $x(\cdot)$ is continuous at $\underline{v}$, and the allocation in $[0,\underline{v}]$ is given as:
    $$x(v) = \left(\frac{v}{\underline{v}}\right)^{\frac{1}{M-1}}x(\underline{v}), \quad \forall v \in [0, \underline{v}].$$
    Next, combined with $p(\bar{v})= \bar{v} x(\bar{v})$, we have 
    $M \myerson(\bar{v}) - \bar{v}x(\bar{v}) = M \myerson(\underline{v}) - \underline{v}x(\underline{v})$,
    that is,
    \begin{equation}
    \label{eq:equal_underline_bar}
        x(\bar{v}) = \frac{M\bar{v} - \underline{v}}{(M-1)\bar{v}} \cdot x(\underline{v}).
    \end{equation}
    
    In order to argue that allocation rule $x(\cdot)$ is not revenue-maximizing, we construct a new allocation rule as
    \begin{equation}
    \label{eq:bar_x}
            \bar{x}(v) = \left\{
            \begin{array}{lr}
                x(v) & \text{if }  v \geq \bar{v} \\
                \displaystyle \left(\frac{v}{\bar{v}}\right)^{\frac{1}{M-1}}x(\bar{v}) & \quad\text{if } v \in [0, \bar{v}).
            \end{array}
            \right.
        \end{equation}
    That is, we replace the first price interval $[0, \underline{v}]$ and the flat interval $[\underline{v}, \bar{v}]$ in $x(\cdot)$ with a single first-price interval $[0, \bar{v}]$ in $\bar{x}(\cdot)$. 
    Comparing the two allocation rules $x(\cdot)$ and $\bar{x}(\cdot)$, 
    we first note from Lemma~\ref{lem:pricing_opt} and Lemma~\ref{lem:TwoPossibility} that 
    $p(\bar{v}) = M \myerson(\bar{v}) = \bar{v} x(\bar{v}) = \bar{v} \bar{x}(\bar{v})$,
    which indicates that
    \begin{equation}
    \label{eq:SameArea}
        \int_{0}^{\bar{v}} x(z) \diff z  = \int_{0}^{\bar{v}} \bar{x}(z) \diff z,
    \end{equation}
    because both sides equal to $\bar{v}x(\bar{v})(M-1)/M$.
    Next, we have for $v\in [0, \underline{v}]$,
    \begin{equation*}
        \begin{aligned}
            \bar{x}(v) - x(v) &= \left(
            \frac{v}{\bar{v}}\right)^{\frac{1}{M-1}}x(\bar{v})
            - \left(\frac{v}{\underline{v}}\right)^{\frac{1}{M-1}}x(\underline{v})\\
          & = v^{\frac{1}{M-1}}\cdot x(\underline{v})\cdot \left( \left(\frac{1}{\underline{v}}\right)^{\frac{1}{M-1}} - \left(\frac{1}{\bar{v}}\right)^{\frac{1}{M-1}} \cdot \frac{M\bar{v} - \underline{v}}{(M-1)\bar{v}} \right).
        \end{aligned}
    \end{equation*}
    Since $v$ only appears in the first term, $\bar{x}(v) - x(v)$ must be constantly positive or negative in $(0, \underline{v}]$, determined by the last term. Combined with~\eqref{eq:SameArea} and $x'(v)=0, \forall v\in(\underline{v}, \bar{v})$, we know it is constantly negative, \emph{i.e.}, $\bar{x}(v)<x(v)$ for all $v\in (0,\underline{v}]$. Therefore, there exists a threshold $v^*\in (\underline{v}, \bar{v})$ such that $x(v)> \bar{x}(v)$ for all $v\in [0, v^*)$ and $x(v)\leq \bar{x}(v)$ for all $v\in [v^*, \bar{v})$. 
    Combining this with Equation~\eqref{eq:SameArea}, we have
    \begin{equation}
    \label{eq:SameArea_2}
        \int_0^{v^*} (x(z) - \bar{x}(z)) \diff z = \int_{v^*}^{\bar{v}} (\bar{x}(z) - x(z)) \diff z > 0.
    \end{equation}    
    Next, for any allocation rule $x(\cdot)$, we denote
    \begin{equation*}
        \rev_{[0,\bar{v}]}^{x(\cdot)} = M \int_{0}^{\bar{v}} \psi(z) x(z) \diff z,
    \end{equation*}
    and we can compare the revenue generated from $x(\cdot)$ and $\bar{x}(\cdot)$ in the interval $[0, \bar{v}]$:
    \begin{equation*}
        \begin{aligned}
            \rev_{[0,\bar{v}]}^{x(\cdot)} - \rev_{[0,\bar{v}]}^{\bar{x}(\cdot)} & = M \left( \int_{0}^{v^*} \psi(z) (x(z) - \bar{x}(z)) \diff z - \int_{v^*}^{\bar{v}} \psi(z) ( \bar{x}(z) - x(z)) \diff z \right).
        \end{aligned}
    \end{equation*}
Finally, we can see that this difference is always negative due to Equation~\eqref{eq:SameArea_2} and the fact that $\psi(\cdot)$ is non-decreasing.\footnote{We omit an ill-defined special case where $\psi(v)=0, \forall v\in[0,\bar{v}]$, since in such case $\psi(v)$ will be infinity when $v\rightarrow 0$.} Fig.~\ref{fig:same_area} demonstrates the idea of this argument.
    
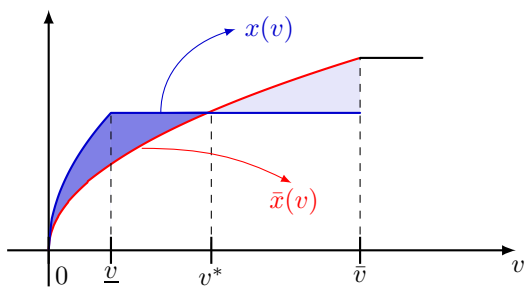
\begin{figure}
\centering

\begin{tikzpicture}[yscale=2, xscale=2.3]
  \def\textscale{1}
  \def\xmax{2.7} 
  \def\ymax{1.6} 
  \def\kx{0.8*\xmax}
  \def\ky{0.8*\ymax} 
  \def\kys{1.0/1.4*\ky}
  \def\a{1.0/6.0*\kx} 
  \def\b{5.0/6.0*\kx} 
  \fill[myblue!50] (\a,0) -- plot[domain=0:\a] (\x, {(\x/(\a))^0.5*\kys});
  \fill[myblue!10] (\b,\kys) -- plot[domain=0.94/2.7*\xmax:\b] (\x, {(\x/(\b))^0.5*\ky});
  \fill[myblue!50] (\a,\kys) -- plot[domain=\a:0.94/2.7*\xmax] (\x, {(\x/(\b))^0.5*\ky});
  \fill[white] (\a,0) -- plot[domain=0:\a] (\x, {(\x/(\b))^0.5*\ky});
  \draw[->,thick] (0,-0.15*\ymax) -- (0,\ymax);
  \draw[->,thick] (-0.15*\ymax,0) -- (\xmax,0) node[right=1,below,scale=\textscale] {$v$};
  \draw[yline,domain =0:\b,smooth,variable=\x,samples=500] plot (\x, {(\x/(\b))^0.5*\ky});
  \draw[xline,domain =0:\a,smooth,variable=\x] plot (\x, {(\x/(\a))^0.5*\kys});
  \draw[xline,line cap=round] (\a,\kys) -- (\b,\kys);
  \draw[xline,color=black,line cap=round] (\b,\ky) -- (\kx,\ky);
  \draw[dashed] (\a,0) -- (\a,\kys);
  \draw[dashed] (\b,0) -- (\b,\ky);
  \draw[dashed] (0.94/2.7*\xmax,0) -- (0.94/2.7*\xmax,\kys);
  \draw[anotline,color=red,line cap=round] (0.25*\kx,0.53*\ky) to [out=0,in=150] (0.65*\kx,0.35*\ky) node[left=-1,below=-1.7,scale=\textscale] {$\bar{x}(v)$};
  \draw[anotline,color=myblue,line cap=round] (0.3*\kx,\kys) to [out=90,in=200] (0.5*\kx,1.15*\ky) node[right=0,scale=\textscale] {$x(v)$};
  \tick{\a,0}{90} node[below=-2.5,scale=\textscale] {$\underline{v}$};
  \tick{\b,0}{90} node[below=-2.5,scale=\textscale] {$\bar{v}$};
  \tick{0.94/2.7*\xmax,0}{90} node[below=-2.5,scale=\textscale] {$v^*$};
  \tick{0,0}{90} node[left=-5,below=-2.5,scale=\textscale] {$0$};  
\end{tikzpicture}
\caption{An illustration for the proof of Lemma~\ref{lem:OneInterval} with $M=3$. The blue line denotes an allocation rule $x(\cdot)$ with $x'(v)=0$ in $(\underline{v}, \bar{v})$ where $\bar{v}<v_{\mmax}$, and $x(v)$ for $v\in [0,\underline{v})$ is computed by Proposition~\ref{prop:ExactAllocation}. The red line denotes our constructed allocation rule $\bar{x}(\cdot)$ as given in~\eqref{eq:bar_x}. The point $v^*$ denotes the intersection of  $x(\cdot)$ and $\bar{x}(\cdot)$. By Equation~\eqref{eq:SameArea_2} we have the areas of the two shadowed regions are the same. Furthermore, as $\psi(\cdot)$ is non-decreasing, we can conclude that $\bar{x}(\cdot)$ leads to a higher revenue than $x(\cdot)$.}
\label{fig:same_area}
\end{figure}

Therefore, we conclude that $\bar{x}(\cdot)$ generates a higher revenue than $x(\cdot)$, contradicting the fact that $(x, p)$ is an optimal auction. This completes the proof.
\end{proof}

We now proceed to complete the proof of Theorem~\ref{thm:opt}.
\begin{proof}[Proof of Theorem~\ref{thm:opt}]
    By Lemma~\ref{lem:OneInterval}, we know $x'(v)>0$ for all valuations $v\in (0,D)$ and $x'(v)=0$ for all valuations $v\in (D,v_{\mmax})$.
    First, we have $x(D)=1$, since otherwise, the revenue could be improved by setting $x(v) = 1$ for all  $v\in[D,v_{\mmax}]$.
    Next, we find the optimal threshold $D^*$ that maximizes the overall revenue.
    By Proposition~\ref{prop:ExactAllocation}, the allocation rule for valuations $v\in [0, D]$ is
    $$x(v)=\left(\frac{v}{D}\right)^{\frac{1}{M-1}} \cdot x(D) = \left(\frac{v}{D}\right)^{\frac{1}{M-1}}, \forall v\in [0, D].$$
    And the overall revenue is 
    $$\rev=\int_{0}^{D} \psi(v) \left(\frac{v}{D}\right)^{\frac{1}{M-1}}  \diff v+\int_{D}^{v_{\mmax}}\psi(v) \diff v.$$
    In order to maximize this revenue, we compute its derivative 
    \begin{equation}
    \label{eq:RevDerivation}
        \frac{\diff \rev}{\diff D} = -\frac{1}{M-1}\cdot D^{-\frac{M}{M-1}}\cdot \int_0^{D}\psi(v) v^{\frac{1}{M-1}} \diff v.
    \end{equation}
    Looking at this derivative, we see that the term $-\frac{1}{M-1}\cdot D^{-\frac{M}{M-1}}$ is always negative and $v^{\frac{1}{M-1}}$ is always positive. Furthermore, we have $\psi(0)=-1<0$ and  $\psi(v_{\mmax})>0$. By the DMR condition, $\psi(\cdot)$ is non-decreasing. Therefore, we have the following two cases:
    \begin{itemize}
        \item If $\int_0^{v_{\mmax}}\psi(v)v^{\frac{1}{M-1}} \diff v>0$, then by letting $\int_0^{D^{*}}\psi(v)v^{\frac{1}{M-1}}\diff v=0$, the revenue will increase with $D$ until $D^{*}$ and then decrease. Therefore, the optimal solution is achieved at $D= D^{*}$.
        \item If $\int_0^{v_{\mmax}}\psi(v)v^{\frac{1}{M-1}} \diff v\leq 0$, then the revenue will always increase with $D$ until $v_{\mmax}$. Therefore, the optimal solution is achieved at $D= v_{\mmax}$.
    \end{itemize}
    This completes the proof.
\end{proof}

\section{Conclusion}
\label{sec8}
In this paper we discuss optimal auction design for bidders who have ex post ROI constraints. We provide characterizations for DSIC auctions and optimal auctions with a single bidder in this setting. We show that the optimal auction may entail a randomized allocation scheme even in the simple single-bidder setting. 

There are several important open questions left in this model. The first and foremost one is to characterize the optimal auction with a single item and multiple ROI-constrained bidders. As we have discussed in the paper, this would require us to go beyond the virtual welfare maximization regime in the Myerson auction setting. We believe such characterization could shed light on the mechanism design for bidders with non-quasilinear utility functions and provide useful insights for practical applications such as online advertising.
Another direction is to study ex post ROI constraints when the target ratio $M_i$ is also private information of bidder $i$. This brings the problem to the domain of multidimensional mechanism design, which is often challenging in the mechanism design literature. Here one possible approach is to identify conditions under which some ``simple'' and deterministic auctions are optimal or close to optimal.

\section{Acknowledgement}
We thank Hu Fu, Yiqing Wang, Yidan Xing, Xiangyu Liu and anonymous reviewers for their insightful and helpful suggestions.
This work was supported in part by National Key R\&D Program of China No. 2021YFF0900800, in part by China NSF grant No. 62220106004, 62322206, 62132018, 62272307, 61972254, 62025204, 62302267, in part by the Natural Science Foundation of Shandong (No. ZR202211150156, ZR2021LZH006), in part by Alibaba Group through Alibaba Innovative Research Program, and in part by Tencent Rhino Bird Key Research Project. The opinions, findings, conclusions, and recommendations expressed in this paper are those of the authors and do not necessarily reflect the views of the funding agencies or the government.

%
%
%
\bibliographystyle{splncs04}
\bibliography{mybibliography}

\appendix


\newpage

\section{Proof of Proposition~\ref{prop:observation}}
\label{appen:prop_1}
\begin{proof}
    First, we have $M_i \myerson_i(0) - 0\cdot x_i(0)=0$ since $x_i(0)=0$ and $p_i(0)=0$. Then we can observe that 
    $$\max_{0 \leq z \leq v}\{M_i \myerson_i(z) - z x_i(z)\}\geq M_i \myerson_i(0) - 0\cdot x_i(0) = 0.$$
    Therefore, $p_i(v) \leq M_i \myerson_i(v)$. Similarly, we have 
    $$\max_{0 \leq z \leq v}\{M_i \myerson_i(z) - z x_i(z)\}\geq M_i \myerson_i(v) - v x_i(v),$$
    hence 
    $$p_i(v)\leq M_i \myerson_i(v) - (M_i \myerson_i(v) - v x_i(v)) = v x_i(v).$$

    For the second statement, we prove that for any values $v<v'\in [0,v_{max}]$, $p_i(v)\leq p_i(v')$. we consider the following two cases:
    \begin{itemize}
        \item When $$\max_{0 \leq z \leq v}\{M_i \myerson_i(z) - z x_i(z)\} = \max_{0 \leq z \leq v'}\{M_i \myerson_i(z) - z x_i(z)\},$$
        we have $$p_i(v') - p_i(v) = M_i (\myerson_i(v')-\myerson_i(v))\geq 0.$$
        \item When $$\max_{0 \leq z \leq v}\{M_i \myerson_i(z) - z x_i(z)\} < \max_{0 \leq z \leq v'}\{M_i \myerson_i(z) - z x_i(z)\},$$
        we denote 
        $$v^* = \argmax_{0 \leq z \leq v'}\{M_i \myerson_i(z) - z x_i(z)\}.$$
        Then we have that $$p_i(v^*) =  M_i \myerson_i(v^*) - (M_i \myerson_i(v^*) - v^* x_i(v^*)) = v^* x_i(v^*),$$
        and 
        $$\max_{0 \leq z \leq v^*}\{M_i \myerson_i(z) - z x_i(z)\} = \max_{0 \leq z \leq v'}\{M_i \myerson_i(z) - z x_i(z)\}.$$
        By the analysis of the above case, we have $$p_i(v')\geq p_i(v^*) =v^* x_i(v^*)> v x_i(v) \geq p_i(v).$$
    \end{itemize}
\end{proof}

\section{Proof of Monotone Allocation Rule (Step 1 of Theorem~\ref{thm:characterization})}
\label{appen:step_1}
\begin{lem}[Monotone Allocation Rule]
 \label{lem:monoto}
 For ex post ROI-constrained bidders, if a mechanism is DSIC, then the allocation rule is monotonically non-decreasing, that is, $x_i(v)\leq x_i(v')$ for all $v<v'$ and bidder $i$.
 \end{lem}

Before the proof, we first present elementary inequalities derived from the property of truthfulness for ex post ROI-constrained bidders.

Given values $v$, $v'$ with $v<v'$, by the truthfulness requirement, when a bidder $i$ with value $v$ misreports $v'$, we have 
\begin{equation}
\label{eq3}
    M_i v x_i(v') - p_i(v')  \leq M_i v x_i(v) - p_i(v),
\end{equation}
or 
\begin{equation}
\label{eq4}
    p_i(v')>v x_i(v').
\end{equation}
Similarly, considering bidder $i$ with value $v'$ misreports $v$, we obtain that 
\begin{equation}
\label{eq1}
    M_i v' x_i(v) - p_i(v) \leq M_i v' x_i(v') - p_i(v'),
\end{equation}
or 
\begin{equation}
\label{eq2}
    p_i(v)>v' x_i(v).
\end{equation}
One can observe that (\ref{eq2}) contradicts the IR requirement since $p_i(v)\leq v x_i(v)< v' x_i(v)$. Hence, in a truthful mechanism, (\ref{eq1}) must hold, and at least one of (\ref{eq3}) and (\ref{eq4}) holds. On the basis of this elementary analysis, we now prove Lemma~\ref{lem:monoto}.

\begin{proof}
We prove the monotonicity by contradiction. Let $v<v'$, we assume $x_i(v)>x_i(v')$ in a truthful mechanism. Then we prove that (\ref{eq3}) and (\ref{eq1}) are not compatible. By (\ref{eq1}), we get
\begin{equation}
\label{eq5}
     p_i(v) - p_i(v')  \geq M_i v' \cdot (x_i(v) - x_i(v')).
\end{equation}
Similarly, by  (\ref{eq3}), we obtain
\begin{equation}
\label{eq6}
     p_i(v) - p_i(v')  \leq M_i v \cdot (x_i(v) - x_i(v')).
\end{equation}
Combining (\ref{eq5}) and (\ref{eq6}), we have $M_i v' \cdot (x_i(v) - x_i(v'))\leq M_i v \cdot (x_i(v) - x_i(v'))$. However, since $x_i(v) > x_i(v')$ and $v'>v$, we can derive a contradiction. 

Next, we continue to prove that (\ref{eq4}) and (\ref{eq1}) are not compatible. By the IR property, we have 
\begin{equation}
\label{eq7}
     p_i(v)\leq v x_i(v).
\end{equation}
Combining (\ref{eq4}), (\ref{eq1}), and (\ref{eq7}), we obtain 
\begin{equation*}
     M_i v' x_i(v) - v x_i(v)<M_i v' x_i(v') - v x_i(v'),
\end{equation*}
that is, 
\begin{equation}
\label{eq8}
     (M_i v' - v) \cdot x_i(v)<(M_i v' - v) \cdot x_i(v').
\end{equation}
As we have $x_i(v)>x_i(v')$ and $v<v'$, (\ref{eq8}) does not hold. In conclusion, we obtain that a mechanism could not be truthful if $x_i(v)>x_i(v')$, and hence, the allocation function $x_i(v)$ must be monotonously non-decreasing.
\end{proof}

 \section{Proof of Uniqueness of Payment  (Step 3 of Theorem~\ref{thm:characterization})}
\label{appen:step_3}
\begin{lem}[Uniquessness of Payment]
    Given any monotonically non-decreasing allocation rule $x(\cdot)$ and $p(0)=0$, there exists a unique payment rule $p(\cdot)$ such that $(x, p)$ is DSIC.
\end{lem}
\begin{proof}
    For any bidder $i$ and any $v_i \in [0, v_{max}]$. We consider two valuations $v_i$ and $w_i = v_i + \delta$ for some sufficiently small $\delta > 0$.
    Following the payment sandwich inequality techniques used in the original Myerson's proof, we have
    $$u(v_i, v_i) \geq u(w_i, v_i) \textrm{ and } u(w_i, w_i) \geq u(w_i, v_i).$$
    This gives us two inequalities that can together sandwich the payment $p_i(v)$.

    If we follow Myerson's original analysis and replace each utility with the value minus the payment for the bidder, these two inequalities would give us
    \[M_i v_i (x_i(w_i)-x_i(v_i)) \leq p_i(w_i) - p_i(v_i) \leq M_i w_i (x_i(w_i)-x_i(v_i)).\]
    In the limit, as $\delta$ approaches 0, we can divide each side by $\delta$ and get
    \begin{equation}
        \label{eq:price_delta}
        \frac{\diff p_i(v_i)}{\diff v_i} = M_i v_i \frac{\diff x_i(v_i)}{\diff v_i}.
    \end{equation}
    The uniqueness of the payment rule then follows by integrating \eqref{eq:price_delta} from $0$ to each value $v$. However, this argument has a problem in the ROI-constrained setting: because of the extra factor $M_i$ in the payment function due to the ROI condition, it is possible that at some point the payment grows to be larger than the bidder's obtained value, therefore violating the IR condition. 


    To deal with this issue, we first note that $p(v_i) \leq v_i \cdot x_i(v_i) \leq w_i \cdot x_i(v_i)$. This means when bidder $i$ has value $w_i$ and misreports her value as $v_i$, her payment will never violate the IR condition. This means we only need to discuss the case when the bidder $i$ with value $v_i$ misreports to $w_i$. There are two cases to discuss.
    \begin{itemize}
        \item If there exists a neighborhood of $v_i$, such that by applying \eqref{eq:price_delta} to $p_i(v_i)$ to get $p_i(w_i)$, we have $p_i(w_i) \leq v_i \cdot x_i(w_i)$. This means in this small neighborhood, the other direction of misreporting will also not be affected by the IR condition. Myerson's analysis can go through, and the derivative of $p_i(\cdot)$ remains fixed and unique at point $v_i$.
        \item If for any small neighborhood of $v_i$, applying~\eqref{eq:price_delta} gives $p_i(w_i) > v_i \cdot x_i(w_i)$. This means bidder $i$ with value $v_i$ would not want to misreport her bid as $w_i$ because it would violate the IR condition for her. Here we notice
        $v_i x_i(w_i) < p_i(w_i) \leq w_i x_i(w_i)$
        must hold. When $\delta$ is sufficiently small, it implies that we must have $p_i(w_i) = w_i x_i(w_i)$.
        This suggests that even though we can no longer apply equation~\eqref{eq:price_delta} in this case. $p_i(w_i)$ is still a unique and fixed value.
    \end{itemize}
    Combining the two cases together, we know that at each point $v \in [0, v_{max}]$, $p_i(v)$ is always unique, therefore proving this claim. 
\end{proof}

\section{Proof of Lemma~\ref{lem:pricing_opt}}
\label{appen:lem_pricing_opt}

\begin{proof}
    We prove this statement by contradiction. Suppose the claim is not true at valuation $v^*$, then we have $v^* = \argmax_{0 \leq z \leq v^*} \{M \myerson(z) - z x(z)\}$ (if there are multiple such $v^*$, choose the smallest one). 
    This means for all $0 \leq v < v^*$, the difference $d(v, v^*)$ between $M \myerson(v^*) - v^* x(v^*)$ and $M \myerson(v) - v x(v)$ is positive:
    \begin{equation}
    \label{eq:d}
        \begin{aligned}
            d(v, v^*) & = (M \myerson(v^*) - v^* x(v^*)) - (M \myerson(v) - v x(v))  \\
            & = M \left(v^* x(v^*) - v x(v) - \int_{v}^{v^*} x(z) \diff z \right) - (v^* x(v^*) - vx(v)) \\
            & = (M-1) (v^* x(v^*) - vx(v)) - M \int_{v}^{v^*} x(z) \diff z >0.
        \end{aligned}
    \end{equation}
    
    Given allocation $x(\cdot)$, our plan is to construct a new monotone allocation rule, which together with the corresponding payment rule can generate a higher revenue. More specifically, we select
    a sufficiently small $\delta>0$ such that letting $d(v, v^*) = 0$ by increasing the allocation $x(v)$ in the interval $(v^* - \delta, v^*)$ does not break the monotonicity of the allocation rule.
    Then, the new allocation rule $\bar{x}(\cdot)$ is defined as 
    \[
    \bar{x}(v) = \left\{
    \begin{array}{lr}
        x(v) & \text{if } v \leq v^* - \delta \quad \text{or} \quad v \geq v^* \\
        \displaystyle \left(\frac{v}{v^*}\right)^{\frac{1}{M-1}}x(v^*) & \text{if } v \in (v^*-\delta, v^*).
    \end{array}
    \right.
    \]

    This way, for all values $v \in (v^*-\delta, v^*)$, with respect to $\bar{x}(\cdot)$, we have
    \begin{equation}
    \label{eq50}
        \begin{aligned}
            d(v, v^*)  = & (M-1) (v^* \bar{x}(v^*) - v\bar{x}(v)) - M \int_{v}^{v^*} \bar{x}(z) \diff z  \\
            =&  (M-1) \left(v^* x(v^*) - v \left(\frac{v}{v^*}\right)^{\frac{1}{M-1}}x(v^*)\right) - M \int_{v}^{v^*} \left(\frac{z}{v^*}\right)^{\frac{1}{M-1}}x(v^*) \diff z \\
            =&  x(v^*) (M-1) \left(v^* - v \left(\frac{v}{v^*}\right)^{\frac{1}{M-1}}\right) - x(v^*) M \int_{v}^{v^*} \left(\frac{z}{v^*}\right)^{\frac{1}{M-1}} \diff z \\
            =&  x(v^*) (M-1) \left(\frac{1}{v^*}\right)^{\frac{1}{M-1}}\left(v^{*\frac{M}{M-1}} - v^{\frac{M}{M-1}}\right) \\
            &- x(v^*) M \left(\frac{1}{v^*}\right)^{\frac{1}{M-1}} \int_{v}^{v^*} z^{\frac{1}{M-1}} \diff z \\
             =& 0.
        \end{aligned}
    \end{equation} 
    Then, we prove that $\bar{x}(v) > x(v)$ for all values $v \in (v^* - \delta, v)$. Assume by contradiction that for any $\delta > 0$, there exists some $v \in (v^*-\delta, v^*)$ such that $\bar{x}(v) \leq x(v)$. Because we have assumed that both $\bar{x}(\cdot)$ and $x(\cdot)$ are right-differentiable and only have finite non-differentiable points, this means there must exist some $\delta' > 0$, such that $\bar{x}(v) \leq x(v)$ holds for all $v \in (v^*-\delta', v^*)$.
    Pick an arbitrary $v$ in this interval. We then have 
    \begin{equation*}
    \begin{aligned}
        d(v, v^*) & = (M-1) (v^* x(v^*) - vx(v)) - M \int_{v}^{v^*} x(z) \diff z \\
        & \leq (M-1) (v^* x(v^*) - v\bar{x}(v)) - M \int_{v}^{v^*} \bar{x}(z) \diff z \\
        & = (M-1) (v^* \bar{x}(v^*) - v\bar{x}(v)) - M \int_{v}^{v^*} \bar{x}(z) \diff z \\
        & =0,
    \end{aligned}        
    \end{equation*}
    where the first equality is by Equation~(\ref{eq:d}) and the last equality is by Equation~(\ref{eq50}). This is a direct contradiction to our previous claim that $d(v, v^*) > 0$ for all $v < v^*$.
    Therefore, we obtain that $\bar{x}(\cdot)$ remains non-decreasing because we only modify $x$ in the interval of $(v^* - \delta, v^*)$ and we have: (1) $\bar{x}(v) > x(v)$ for all values $v \in (v^* - \delta, v^*)$; (2) $\bar{x}(v)$ is increasing in $(v^* - \delta, v^*)$; and (3) $\bar{x}(v)$ is continuous at $v^*$. 

    Finally, we analyze the revenue generated by $\bar{x}(\cdot)$. Let $q(\cdot)$ be the corresponding payment rule of $\bar{x}(\cdot)$ derived from Theorem~\ref{thm:characterization}. We compare $q(v)$ and $p(v)$ at each value $v$.
    \begin{itemize}
        \item When $v \leq v^* - \delta$, we have $q(v) = p(v)$. This is because the allocation remains unchanged in the interval $[0, v^*-\delta)$, and the payment of a bid at value $v$ only depends on the allocation at interval $[0, v]$.
        \item When $v \in (v^* - \delta, v^*)$, we always have $v \in \argmax_{0 \leq z \leq v} \{M \myerson(z) - z \bar{x}(z)\}$, which means the payment of $q(v)$ reduces to the first price, and we have $q(v) = v \bar{x}(v) > v x(v) \geq p(v)$.
        \item When $v \geq v^*$, we again have $q(v) = p(v)$. This is because when $v \geq v^*$, we have $v^* \in \argmax_{0 \leq z \leq v} \{M \myerson(z) - z x(z)\}$ for both $(x, p)$ and $(\bar{x}, q)$. Then the payment formulation of~\eqref{eq:payment} reduces to 
        \[
        p(v) = M\myerson(v) - (M\myerson(v^*) - v^* x(v^*)) = v^* x(v^*) + M(\myerson(v) - \myerson(v^*)).
        \]
        This payment only depends on the allocation at interval $[v^*, v]$, which again $x(\cdot)$ and $\bar{x}(\cdot)$ agree on.
    \end{itemize}
    Combining these cases together, we have $\int_{v}q(v)f(v) \diff v> \int_v p(v)f(v) \diff v$. That is, the new mechanism $(\bar{x}, q)$ has a higher revenue compared to $(x, p)$. 
    This is a contradiction, and this completes the proof.
\end{proof}

\section{Proof of Lemma~\ref{lem:TwoPossibility}}
\label{appen:lem_TwoPossibility}

\begin{proof}
Assume that in some revenue-maximizing auction $(x, p)$, there exists a value $v$ such that the two conditions claimed in the lemma do not hold. That is, we have $p(v) < v x(v)$, and $x'(v) > 0$ or $x'(v)$ does not exist.
Since $x(\cdot)$ is right-differentiable with finite non-differentiable points, by Theorem~\ref{thm:characterization}, we have $p(\cdot)$  is also continuous at all but a finite number of points in its domains. This means we can always find an interval
$[\underline{v},\bar{v}]$, such that $p(v) < v x(v)$ and $x'_i(v)> 0$ for all valuations $v \in [\underline{v},\bar{v}]$.

Recall that when $p(v) < v x(v)$ for all $v$, by Lemma~\ref{lem:pricing_opt}, the payment reduces to $p(v) = M \myerson(v)$ and we can write the revenue of the mechanism as $\rev = M \int_{0}^{v_{\mmax}} \phi(z) x(z) f(z) \diff z$.
Next, we look at the virtual values $\phi(v)$ within this interval $[\underline{v},\bar{v}]$. Our plan is to modify the allocation $x(v)$ in (a subinterval of) this interval based on the sign of $\phi(v)$ while maintaining the monotonicity of $x(\cdot)$ and $p(v) < v x(v)$ in the interval, and the expected revenue will (weakly) increase.

We consider the following three cases:
\begin{itemize}
    \item There exists an interval $[a,b] \subseteq [\underline{v},\bar{v}]$ such that $\phi(v)>0, \forall v\in[a,b]$. In this case, we define 
    \[
    \bar{x}(v) = \left\{
    \begin{array}{lr}
        x(v) & \text{if } v \leq a \quad \text{or} \quad v \geq b \\
        \displaystyle x(v) + \delta \cdot \frac{b-v}{b-a} & \text{if } v \in (a, b),
    \end{array}
    \right.
    \]
    where $\delta > 0$ is sufficiently small such that:
    \begin{enumerate}
        \item $\bar{x}(\cdot)$ is still non-decreasing; and
        \item $p(v) < v \bar{x}(v)$ still holds in the interval $[\underline{v}, \bar{v}]$, which means the corresponding payment is still $p(v) = M \myerson(v)$ in this interval.
    \end{enumerate}
    Conditions (1) and (2) imply we can still write the expected revenue as 
    \begin{equation*}
        \rev = M \int_{0}^{v_{\mmax}} \phi(z) \bar{x}(z) f(z) \diff z
    \end{equation*}
    for the new mechanism with allocation $\bar{x}$.
    In the meanwhile, $\bar{x}(v)$ is point-wise larger than $x(v)$ at all values $v \in [a, b]$ where $\phi(v)$ is always positive, and outside this interval the two allocations remain the same. This means $\bar{x}(\cdot)$, together with its corresponding payment rule, would yield strictly higher revenue than the previous mechanism. A contradiction.
  
    \item There exists an interval $[a,b]$ in  $[\underline{v},\bar{v}]$ such that $\phi(v)<0, \forall v\in[a,b]$. Similar to the first case, we define 
    \[
    \bar{x}(v) = \left\{
    \begin{array}{lr}
        x(v) & \quad  \text{if } v \leq a \quad \text{or} \quad v \geq b \\
        \displaystyle x(v) - \delta \cdot \frac{v-a}{b-a} & \quad \text{if } v \in (a, b),
    \end{array}
    \right.
    \]
    where $\delta>0$ is sufficiently small such that $\bar{x}(\cdot)$ is still non-decreasing and $p(v) < v x(v)$ still holds in the interval $[\underline{v}, \bar{v}]$.
    Using the same argument as in the previous case, we can again argue that $\bar{x}(\cdot)$ gives a higher revenue than $x(\cdot)$. Again a contradiction.
    \item If neither of the first two cases happens, we must have $\phi(v)=0, \forall v\in [\underline{v}, \bar{v}]$. In this case, as long as the monotonicity of $x(\cdot)$ and $p(v) \leq v x(v)$ is maintained in the interval, any modification of $x(\cdot)$ would generate the same revenue. Therefore we can always find an allocation $\bar{x}(\cdot)$ that satisfies one of the required two conditions and have the same revenue as that of $x(\cdot)$. Therefore $\bar{x}(\cdot)$ still gives a revenue-maximizing auction. 
\end{itemize}
This concludes the proof.
\end{proof}

\end{document}